\documentclass[10pt]{article}
\pdfoutput=1

\usepackage{xparse}
\usepackage{etex}
\usepackage{etoolbox}

\usepackage[a4paper]{geometry}

\usepackage[T1]{fontenc}
\usepackage[utf8]{inputenc}
\usepackage{microtype}

\usepackage{abstract}

\usepackage[british]{babel}
\usepackage{csquotes}
\usepackage[all]{foreign}
\defasforeign[mutatismutandis]{mutatis mutandis}
\defasforeign[Mutatismutandis]{Mutatis mutandis}
\defasforeign[criterium]{criterium}
\defnotforeign[wrt]{{\xperiodafter{w.r.t}}}

\usepackage{amsmath}
\usepackage{amsthm}
\usepackage{amsfonts,amssymb,bm,stmaryrd}
\usepackage{mathtools}
\usepackage{esvect}
\renewcommand{\vec}{\vv}

\allowdisplaybreaks

\usepackage{thmtools}

\usepackage{float}
\usepackage{caption}
\usepackage{subcaption}

\usepackage{enumitem}

\usepackage{tikz}
\usetikzlibrary{calc,arrows,positioning,intersections}

\usepackage{hyperref}
\usepackage{nameref}

\usepackage[capitalise,noabbrev,nameinlink]{cleveref}

\hypersetup{
	hidelinks,
	final,
}

\usepackage[
	backend=bibtex,
	style=numeric-comp,
	maxbibnames=99, %
	giveninits=true,
	sortcase=false,
	sorting=nyt,
	doi=true,
	isbn=true,
	url=true
	]{biblatex}
\theoremstyle{plain}
\newtheorem{theorem}{Theorem}
\newtheorem{corollary}[theorem]{Corollary}
\newtheorem{proposition}[theorem]{Proposition}
\newtheorem{lemma}[theorem]{Lemma}
\theoremstyle{definition}
\newtheorem{definition}{Definition}

\newtheorem*{notation}{Notation}
\newtheorem{remark}{Remark}
\newtheorem{example}[remark]{Example}

\makeatletter
\newcommand{\footnoteref}[1]{%
	\protected@xdef\@thefnmark{\ref{#1}}\@footnotemark%
}
\makeatother

\makeatletter
\newcommand{\Superimpose}[2]{%
	{\ooalign{$#1\@firstoftwo#2$\cr\hfil$#1\@secondoftwo#2$\hfil\cr}}}
\makeatother

\makeatletter
\newcommand{\customlabel}[2]{%
	\protected@write \@auxout {}{\string \newlabel {#1}{{#2}{\thepage}{#2}{#1}{}} }%
	\hypertarget{#1}{#2}
}
\makeatother

\newcommand{\oversetlabel}[3]{
	\phantomsection
	\overset{({\customlabel{#2}{#1}})}{#3}
}

\newcounter{proofstep}
\newcommand*{\stepref}[1]{}
\newcommand*{\steplabel}[2][]{}
\newcommand*{\steptag}[1][]{}
\newcommand*{\setupstepsequence}[2][0]{%
	\renewcommand*{\stepref}[1]{\eqref{#2-##1}}%
	\renewcommand*{\steplabel}[2][\theproofstep]{%
		\stepcounter{proofstep}%
		\oversetlabel{\textit{\roman{proofstep}}}{#2-##1}{##2}%
	}%
	\renewcommand*{\steptag}[1][\theproofstep]{%
		\stepcounter{proofstep}%
		\tag{\textit{\roman{proofstep}}}%
		\label{#2-##1}%
	}%
	\setcounter{proofstep}{#1}%
}

\DeclareMathOperator{\bis}{bis}
\DeclareMathOperator{\supp}{supp}
\DeclareMathOperator{\ffix}{\nu}

\NewDocumentCommand\card{ s m }{
   \IfBooleanTF{#1}{%
     	{|{#2}|}%
   }{%
      {\left|#2\right|}%
   }%
}
\NewDocumentCommand\carr{ s m }{
   \IfBooleanTF{#1}{%
     	{\mathrm{car}({#2})}%
   }{%
      {\mathrm{car}\left({#2}\right)}
   }%
}

\newcommand{\ltrue}{\texttt{t\!t}}
\newcommand{\lfalse}{\texttt{f\!f}}

\DeclareBoldMathCommand{\bflangle}{\langle}
\DeclareBoldMathCommand{\bfrangle}{\rangle}

\newcommand{\cat}[1]{\textnormal{\textsf{#1}}\xspace}
\newcommand{\Set}{\cat{Set}}
\newcommand{\Sys}[1]{{\cat{Sys}(#1)}}

\newcommand{\ff}[1]{\mathcal{F}_{\mspace{-1mu}#1}}
\newcommand{\fp}[1]{\mathcal{P}_{\mspace{-1mu}#1}}
\newcommand{\fpf}{\fp{\mspace{-2mu}f}}

\newcommand{\defeq}{\triangleq}

\newcommand{\dotrel}[1]{\mathrel{\dot{#1}}}

\newcommand{\reducesTo}{\preccurlyeq}
\newcommand{\reducesFrom}{\succcurlyeq}
\newcommand{\reducesEq}{\approxeq}
\newcommand{\freduceTo}{\dotrel{\reducesTo}}

\newcommand{\freduceEq}{\dotrel{\reducesEq}}

\title{On the trade-off between labels and weights in quantitative bisimulation}
\author{
	{Marco Peressotti}\\
	\small University of Southern Denmark\\
	\small \href{mailto:peressotti@imada.sdu.dk}{\tt peressotti@imada.sdu.dk}
}

\date{}

\begin{document}

\maketitle

\begin{abstract}%
\emph{Reductions for transition systems} have been recently introduced as a uniform and principled method for comparing the expressiveness of system models with respect to a range of properties, especially bisimulations. In this paper we study the expressiveness (\wrt bisimulations) of models for quantitative computations such as \emph{weighted labelled transition systems} (WLTSs), \emph{uniform labelled transition systems} (ULTraSs), and \emph{state-to-function transition systems} (FuTSs). We prove that there is a trade-off between labels and weights: at one extreme lays the class of ``unlabelled'' weighted transition systems where information is presented using weights only; at the other lays the class of labelled transition systems (LTSs) where information is shifted on labels. These categories of systems cannot be further reduced in any significant way and subsume all the aforementioned models.
\end{abstract}


\section{Introduction}

\emph{Weighted labelled transition systems} (WLTSs) \cite{klin:sas2009} are a meta-model for systems with quantitative aspects: transitions are of the form $P\xrightarrow{a,w}Q$ and labelled with \emph{weights} $w$ that are taken from a given monoidal weight structure and express the quantity associated to the computational step. Many computational aspects can be captured just by changing the underlying weight structure: weights can model probabilities, resource costs, stochastic rates, \etc; as such, WLTSs are a generalisation of labelled transition systems (LTSs) \cite{milner:cc}, probabilistic systems (PLTSs) \cite{gsb90:ic}, stochastic systems \cite{hillston:pepabook}, among others. Definitions and results developed in this setting instantiate to existing models, thus recovering known results and discovering new ones. In particular, the notion of \emph{weighted bisimulation} in WLTSs coincides with that of (strong) bisimulation for all the aforementioned models \cite{klin:sas2009}.

In the wake of these encouraging results, other meta-models have been proposed aiming to cover an even wider range of computational models and concepts.
\emph{Uniform labelled transition systems} (ULTraSs) \cite{denicola13:ultras} are systems whose transitions have the form $P\xrightarrow{a}\phi$, where $\phi$ is a \emph{weight function} assigning weights to states; hence, ULTraSs can be seen both as a non-deterministic extension of WLTSs and as a generalisation of Segala's probabilistic systems \cite{sl:njc95} (NPLTSs).
In \cite{mp:qapl14,mp:tcs2016} a (coalgebraically derived) notion of bisimulation for ULTraSs is presented and shown to precisely capture bisimulations for weighted and Segala systems.
\emph{Function-to-state transition systems} (FuTSs) were introduced in \cite{denicola13:ustoc} as a generalisation of the above, of IMCs \cite{hermanns:imcbook}, and of other models used to formalise the semantics of several process calculi with quantitative aspects. Later, \cite{latella:lmcs2015} defined a (coalgebraically derived) notion of (strong) bisimulation for FuTSs which instantiates to known bisimulations for all the aforementioned models and hence can be taken as a general schema for defining \emph{quantitative bisimulations} (\cf \cite{latella:qapl2015}).

Given all these meta-models, it is natural to wonder about their \emph{expressiveness}.
We should consider not only the class of systems these frameworks can represent,  but also \emph{whether} these representations are faithful with respect to the properties we are interested in.  
Intuitively, a meta-model $\cat{M}$ is \emph{subsumed by $\cat{M}'$ according to a property $P$} if any system $S$ which is an instance of $\cat{M}$ with the property $P$, is also an instance of $\cat{M}'$ preserving $P$.

In this work we study these meta-models according to their ability to correctly express \emph{strong bisimulation}. In this context, a meta-model $\cat{M}$ is \emph{subsumed} by $\cat{M}'$ if any system $S$ which is an instance of $\cat{M}$, is also an instance of $\cat{M}'$ preserving and reflecting bisimulations.

\begin{figure}%
	\begin{minipage}[b]{.5\linewidth}%
		\centering
		\begin{tikzpicture}[
				auto,
				xscale=1.2, yscale=.8,
				baseline=(current bounding box.center),
			]
			\node (futs)   at (.5,3) {\cat{FuTS}};
			\node (ultras) at (.5,2) {\cat{ULTraS}};
			\node (wlts)   at ( 1,1) {\cat{WLTS}};
			\node[gray] (nplts)  at ( 0,1) {\cat{NPLTS}};
			\node[gray] (lts)    at ( 1,0) {\cat{LTS}};
			\node[gray] (plts)   at ( 0,0) {\cat{PLTS}};

			\node[left of=plts,gray] {\dots};
			\node[right of=lts,gray] {\dots};

			\draw (futs) -- (ultras);
			\draw (ultras) -- (wlts);
			\draw[gray] (ultras) -- (nplts);
			\draw[gray] (nplts) -- (plts);
			\draw[gray] (nplts) -- (lts);
			\draw[gray] (wlts) -- (plts);
			\draw[gray] (wlts) -- (lts);
		\end{tikzpicture}%
		\subcaption{Pre \cite{mp:ictcs2016}.}%
		\label{fig:hierarchy}
	\end{minipage}%
	\begin{minipage}[b]{.5\linewidth}%
		\centering
		\begin{tikzpicture}[
				auto,
				xscale=1.2, yscale=.8,
				baseline=(current bounding box.center),
			]
			\node (wlts) at (.5,2) {$\cat{WLTS}\reducesEq\cat{FuTS}$};
			\node[gray] (nplts)  at (.5,1) {\cat{NPLTS}};
			\node[gray] (lts)    at ( 1,0) {\cat{LTS}};
			\node[gray] (plts)   at ( 0,0) {\cat{PLTS}};
			
			\node[left of=plts,gray] {\dots};
			\node[right of=lts,gray] {\dots};
			
			\draw[gray] (wlts) -- (nplts);
			\draw[gray] (nplts) -- (plts);
			\draw[gray] (nplts) -- (lts);
		\end{tikzpicture}%
		\subcaption{Post \cite{mp:ictcs2016}.}%
		\label{fig:hierarchy-collapsed}%
	\end{minipage}
	\caption{The hierarchy of FuTSs models.}%
\end{figure}
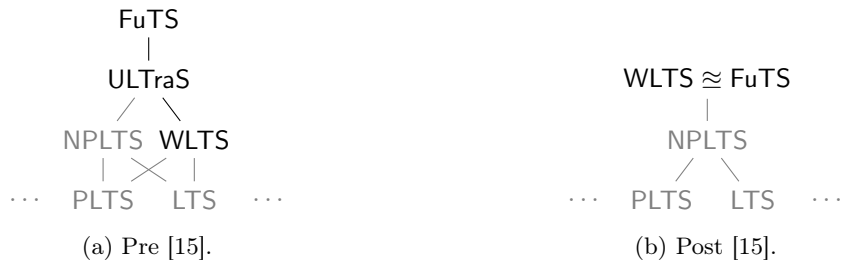
Previous work  \cite{klin:sas2009,denicola13:ultras,mp:qapl14,mp:tcs2016,latella:lmcs2015} has shown that, according to this order, each of the models mentioned above subsumes the previous ones, thus forming the hierarchy shown in \cref{fig:hierarchy}. These results rely on direct inclusions between classes of models (\eg LTSs are WLTSs on booleans) or some \adhoc arguments (\eg WLTSs are ``functional'' ULTraSs). Furthermore, in \loccit it is never shown if any of these models is \emph{strictly} more expressive than the ones it subsumes.

To address this issue in a uniform and principled way, we introduced in \cite{mp:ictcs2016} the notion of \emph{reduction between categories of systems} and used it to formally define expressiveness of models with respect to (strong) bisimulation. 
We remark that this notion is far more general and can be used to study any category of state-based transition systems: all constructions and results  are developed abstracting from the ``type'' of computation under scrutiny thanks to the theory of coalgebras \cite{rutten:tcs2000}. This level of abstraction allowed us to formulate general results for the systematic derivation of reductions. As an application, we derived a reduction taking FuTSs to WLTSs hence proving that the former are subsumed by the latter. Because of this reduction, the hierarchy in \cref{fig:hierarchy} collapses to that in \cref{fig:hierarchy-collapsed}.

In \cite{mp:ictcs2016} we left unanswered a fundamental question about this hierarchy of models: 
\begin{quote}\centering
Can the hierarchy of FuTSs/WLTSs be further collapsed in any meaningful way? 
\end{quote}
To provide this answer we extend the framework of reductions with new results concerning the existence of reductions and determine necessary and sufficient conditions for a wide class of systems of interest. Thanks to these methods, we are able to formally prove that there is a trade-off between labels and weights in the definition of WLTSs and similar models: information and complexity can be shifted from labels to weights and \viceversa in ways that are coherent with model semantics as per the notion of quantitative bisimulation. At the two ends of this trade-off we find labelled transition systems and ``unlabelled'' weighted transition systems. 
As a consequence, LTSs subsume the whole hierarchy of models described above.

\paragraph*{Synopsis}
In \cref{sec:transition-systems} we recall from \cite{rutten:tcs2000} basic notions of coalgebras and their use as a uniform model of (discrete) transition systems, their bisimulation, and final semantics.
In \cref{sec:reductions} we recall from \cite{mp:ictcs2016} relevant definitions about reductions and extend this framework with new results. In particular, in \cref{sec:reductions-final-coalgebras} we introduce necessary and sufficient conditions for reductions to exist under the assumption that all types involved have final coalgebras.
In \cref{sec:reducing-wlts} we focus the lens of reductions on WLTSs and study under which conditions the category of WLTSs can be reduced. 
In \cref{sec:reducing-futs} we apply the results we introduced in \cref{sec:reductions-final-coalgebras} to the category of FuTSs to prove that it reduces to that of WLTSs. The result is known since \cite[Sec.~6]{mp:ictcs2016} but here we provide an alternative proof.
Final remarks are in \cref{sec:conclusions}.

\section{Discrete transition systems}
\label{sec:transition-systems}

Coalgebras are a well established framework for modelling and studying the operational semantics of (abstract) computational devices such as automata, concurrent systems, and reactive ones; the methodology is termed \emph{universal coalgebra} \cite{rutten:tcs2000}. 
In this approach, the first step is to define a \emph{behavioural endofunctor} $T$ over the category $\Set$ of small sets and functions modelling the computational aspects under scrutiny. Here, the term ``modelling'' has to be intended in the sense that, for $X$ a set of states, $TX$ is the set of possible behaviours over $X$. Then, a system is modelled by a $T$-coalgebra \ie a pair $(X,\alpha\colon X\to TX)$ where the set $X$, called \emph{carrier}, is the state-space of the system and the map $\alpha$, called \emph{structure}, associates each state with its behaviour. Coalgebras are often identified with their structure; in this situation the carrier $X$ of a $T$-coalgebra $\alpha\colon X \to TX$ will be denoted by $\carr{\alpha}$.
The definition of the endofunctor $T$ constitutes the crucial step of this method, as it corresponds to specify the behaviours that the systems under scrutiny are meant to exhibit \ie their \emph{dynamics} and \emph{observations}.
Hence, $T$-coalgebras and the systems they model are said to be of \emph{type $T$}. Herein, the terms ``coalgebra'' and ``system'' (of type $T$) will be often treated as synonyms.
Once a behavioural endofunctor is defined, this canonically determines a notion of coalgebra homomorphism \ie structure preserving maps between carriers: a \emph{$T$-homomorphism} from $\alpha\colon X \to TX$ to $\beta\colon Y \to TY$ is a function $f\colon X \to Y$ such that $Tf \circ \alpha = f \circ \beta$.
Categories formed by systems and their homomorphisms will be called \emph{categories of systems} and the category of all $T$-systems will be denoted as $\Sys{T}$.

Because homomorphisms preserve structures (hence system dynamics), their codomain can be thought of as a \emph{refinement system} for their domain in the sense that they can aggregate states with equivalent dynamics. This observation corresponds to the notion of \emph{(kernel) bisimulation}: a relation $R$ on $X$ is bisimulation for a system $\alpha\colon X \to TX$ provided it is the kernel of a function underlying some $T$-homomorphism with domain $\alpha$ \cite{staton:lmcs2011}. This notion of bisimulation generalises Milner's strong bisimulation for LTSs as well as its generalisation to all the models considered in this work. The set of all bisimulations for a $T$-system $\alpha$ will be denoted as $\bis(\alpha)$.

Under mild conditions on the behavioural endofunctor $T$, there exists a \emph{final $T$-system} $\ffix T$ which describes all abstract behaviours of type $T$. Final $T$-systems are final objects in the category $\Sys{T}$ \cite{am:ctcs89,worrell:phdthesis}. Finality means that every $T$-system $\alpha$ has a unique homomorphism $!_h$ into the final $T$-system. The function $\carr{\alpha} \to \carr{\ffix T}$ underlying $!_{\alpha}\colon \alpha \to \ffix T$ uniquely associates each state in $\carr{\alpha}$ with a semantics, called \emph{final semantics}, in the form of an abstract behaviour. Final semantics uniquely associates bisimilar states to the same abstract behaviour. As a consequence, states of final coalgebras can only be bisimilar to themselves. This property is called \emph{strong extensionality} and identifies final homomorphisms with the coinductive proof principle \cite{jr:eatcs1997}. Because of its relation with coinduction, the unique homomorphism from a system $\alpha$ into the final one is also called \emph{coinductive extension} of $\alpha$.

\begin{example}[LTSs]
	It has been shown in \cite{rutten:tcs2000} that:
	\begin{itemize}
		\item image finite LTSs over a set of labels $A$ are $(\fpf-)^A$-systems where $\fpf$ is the finite powerset functor;
		\item the notion of (kernel) bisimulation for $(\fpf-)^A$-systems coincides with Milner's (strong) bisimulation for LTSs;
		\item the endofunctor $(\fpf-)^A$ has final systems.
	\end{itemize}
	For exposition sake, sets of labels are implicitly assumed to not be empty---thus avoiding degenerate systems of type $(\fpf-)^0 \cong 1$. 
	Hereafter, let $\cat{LTS}$ denote the category of all image-finite labelled transition systems and let $\cat{LTS}(A) \defeq \cat{Sys}((\fpf-)^A)$ be its subcategory of systems whose labels range over $A$. 
	\qed
\end{example}

For $(M,+,0)$ an abelian monoid\footnote{An abelian monoid is a set $M$ equipped with an associative and commutative binary operation $+$ and a unit $0$ for $+$; such structure is called trivial when $M$ is a singleton.}, the generalised (finite) multiset functor $\ff{M}$ is the endofunctor over $\Set$ which assigns:
\begin{itemize}
	\item
		to each set $X$ the set $\{ \phi\colon X \to M \mid \supp(\phi) \text{ is finite}\}$ of finitely supported weight functions (the support of $\phi$ is the set $\supp(\phi) = \{ x \mid \phi(x) \neq 0\}$);
	\item
		to each function $f\colon X\to Y$ the map $(\ff{M}f)(\phi) = \lambda y \in Y.\sum_{x:f(x)=y} \phi(x)$. (summation is well defined because $\phi$ is finitely supported).
\end{itemize}
Note that if $M$ is trivial then $\ff{M}X$ is always a singleton whence, monoids of weights are implicitly assumed to not be trivial.
Elements of $\ff{M}X$ will be often presented in \emph{formal sum notation}: for $\phi \in \ff{M}X$ we write $\sum_x \phi(x) \cdot x$ or, given $\supp(\phi) = \{x_1,\ldots,x_n\}$, simply $\sum_{i=1,\ldots,n} \phi(x_i)\cdot x_i$. For instance, $\ff{M} (f)(\phi)$ is formulated as $\sum \phi(x) \cdot f(x)$.

\begin{example}[WLTSs]
	Let $A$ be a non-empty set and $M$ a non-trivial abelian monoid. 
	It has been shown in \cite{klin:sas2009} that:
	\begin{itemize}
		\item
			WLTSs with labels drawn from $A$ and weights drawn from $M$ are $(\ff{M}-)^A$-system;
		\item the notion of (kernel) bisimulation for $(\ff{M}-)^A$-systems  coincides with that of weighted bisimulation;
		\item the endofunctor $(\ff{M}-)^A$ has final systems.
	\end{itemize}
	In the sequel, let $\cat{WLTS}$ be the category of all weighted labelled transition systems and write $\cat{WLTS}(A,M)$ for its subcategory $\Sys{(\ff{M}-)^A}$ formed by $(\ff{M}-)^A$-systems.
	If weights are drawn from the monoid $\mathbb{B} = (\{\ltrue,\lfalse\},\lor,\lfalse)$ of boolean values under disjunction then, $\cat{WLTS}(A,\mathbb{B})$ is (isomorphic to) $\cat{LTS}(A)$ and the associated notions of bisimulation coincide.
	\qed
\end{example}

\begin{example}[ULTraSs]
	\looseness=-1
	Let $A$ be a non-empty set and $M$ a non-trivial abelian monoid. 
	It has been shown in \cite{mp:qapl14,mp:tcs2016} that (image finite) ULTraSs on $A$ and $M$ are $(\fpf\ff{M}-)^A$-systems.
	The notion of bisimulation for $(\fpf\ff{M}-)^A$-systems coincides with that of strong bisimulation as per \cite{mp:qapl14,mp:tcs2016}. The functor $(\fpf\ff{M}-)^A$ admits final systems.
	Let $\cat{ULTraS}$ denote the category of all image-finite ULTraSs
	and $\cat{ULTraS}(A,M)$ its subcategory of systems with labels in $A$ and weights in $M$. 
	WLTSs can be cast to ULTraSs by wrapping each target of their transitions into a singleton (\ie by post-composition with components of the powerset unit).  In \cite{denicola13:ultras} ULTraSs obtained in this way are called  \emph{functional} and, as shown in \cite{mp:qapl14}, the casting operation is faithful to the semantics of WLTSs.
\qed
\end{example}

\begin{example}[FuTSs]
FuTSs are systems for any endofunctor $T$ by the grammar
\[
	T \Coloneqq (S-)^A \mid T \times (S-)^A 
	\qquad
	S \Coloneqq \ff{M} \mid \ff{M} \circ S
\]
where $A$ and $M$ range over (non-empty) sets of labels and (non-trivial) abelian monoids, respectively. Any such endofunctor is equivalently described by
\[\textstyle
	(\ff{\vec{M}}f)^{\vec{A}} 
	\defeq 
	\prod_{i=0}^{n} (\ff{\vec{M}_i}f)^{A_i}
	\qquad
	(\ff{\vec{M}_i}f)^{A_i} 
	\defeq 
	(\ff{M_{i,0}}\dots \ff{M_{i,m_i}}f)^{A_i}
\]
where $\vec{A} = \langle A_0,\dots,A_n\rangle$ is a sequence of (non-empty) sets of labels,
each $\vec{M}_i = \langle M_{i,0},\dots,M_{i,m_i}\rangle$ is a sequence of non-trivial abelian monoids, and $\vec{M} = \langle \vec{M}_0,\dots,\vec{M}_n\rangle$. (Up to minor notational variations, this characterisation can be found in \cite{latella:lmcs2015,mp:tcs2016}.)
The notion of bisimulation for $(\ff{\vec{M}}f)^{\vec{A}}$-systems coincides with that of strong bisimulation for FuTSs described in \cite{latella:lmcs2015}. In \loccit it is shown that every $(\ff{\vec{M}}f)^{\vec{A}}$ has final systems.
For any $\vec{A}$ and $\vec{M}$ as above define $\cat{FuTS}(\vec{A},\vec{M})$ as category $\Sys{\ff{\vec{M}}-)^{\vec{A}}}$. Clearly, $\cat{FuTS}(\langle A),\langle M\rangle)$ and $\cat{FuTS}(\langle A\rangle,\langle \mathbb{B},M\rangle)$ coincide with $\cat{WLTS}(A,M)$ and $\cat{ULTraS}(A,M)$, respectively. As a consequence, the categories $\cat{LTS}$, $\cat{WLTS}$, and $\cat{ULTraS}$ are all subcategories of $\cat{FuTS}$, the category of all FuTSs.
\qed
\end{example}

\section{Reductions for discrete transition systems}
\label{sec:reductions}

In \cite{mp:ictcs2016} the notion of \emph{reduction} was introduced in order to formalise the intuition that a behaviour type is (at least) as expressive as another whenever systems and homomorphisms of the latter can be ``encoded'' as systems and homomorphisms of the former and provided that their semantically relevant structures are both preserved and reflected. In this work we extend the framework of reductions with new results concerning existence of reductions under the assumption that the involved types have final systems.

\subsection{Reductions}
Intuitively, to reduce a transition system (\emph{source}) to another (\emph{target}) of a possibly different type means to ``encode'' the state space of the former into the state space of the latter while preserving and reflecting their semantically relevant structure and properties. Because the source and target system of a reduction may be of different types, structure preservation cannot be formalised as a system homomorphism. Instead, reductions rely on an indirect expression of homomorphisms: bisimulations. Reductions require that every bisimulation for the source system is assigned to bisimulations for the target system that are coherent with the encoding of the state space.

\begin{definition} 
	\label[definition]{def:system-reduction}
	For systems $\alpha$ and $\beta$, a \emph{(system) reduction} $\sigma\colon \alpha \to \beta$ is given by
	 \begin{enumerate}
	 	\item
	 		a function $\sigma^{c}\colon \carr{\alpha} \to \carr{\beta}$
			and
		\item
			a left-total\footnote{A relation $R\subseteq X \times Y$ is called left-total or multivalued function if for every $x \in X$ there is $y \in Y$ s.t.~$x \mathrel{R} y$.} relation $\sigma^{b}\subseteq \bis(\alpha) \times \bis(\beta)$ 
	\end{enumerate}
	such that $\sigma^{c}$ carries a relation homomorphism for any pair of bisimulations in $\sigma^{b}$, \ie:
	\begin{equation}
		\tag{$\diamondsuit$}
		\label{eq:reduction-bisim-condition}
		R \mathrel{\sigma^{b}} R' \implies
		(x \mathrel{R} x' \iff \sigma^{c}(x) \mathrel{R'} \sigma^{c}(x'))\text{.}
	\end{equation}
	A system reduction $\sigma\colon \alpha \to \beta$ is called \emph{full} whenever $\sigma^{c}\colon \carr{\alpha} \to \carr{\beta}$ is surjective.
\end{definition}

As a consequence of Condition \eqref{eq:reduction-bisim-condition}, the function $\sigma^{c}$ is always injective and the correspondence $\sigma^{b}$ is always left-unique\footnote{A relation $R\subseteq X \times Y$ is called left-unique or injective if for every $y \in Y$ there is at most one $x \in X$ s.t.~$x \mathrel{R} y$.}.

\begin{proposition}
	\label[proposition]{thm:sys-reduction-diamond}
	For $\sigma\colon \alpha \to \beta$ a system reduction, 
	$\sigma^{c}$ is injective and $\sigma^{b}$ is left-unique.
\end{proposition}
\begin{proof}
	Observe that the diagonal (or identity) relation on a system state space is always a bisimulation. For a set $X$, write $\Delta_{X}$ for the diagonal on $X$. By left-totality of $\sigma^{b}$, there is a bisimulation $R$ for $\beta$ such that $\Delta_{\carr{\alpha}} \mathrel{\sigma^{b}} R$. Then, it follows from Condition \eqref{eq:reduction-bisim-condition}, that $\sigma^{c}(x) \mathrel{R} \sigma^{c}(x')$ if and only if $x \mathrel{\Delta_{\carr{\alpha}}} x'$ \ie if and only if $x = x'$.
	Left-uniqueness follows immediately from injectivity of $\sigma^{c}$.
\end{proof} 

\noindent 
The first property guarantees that reductions preserve the identity of states and this is of relevance \eg when these are presented as terms of some process calculus or programming language, since reductions provide alternative but faithful semantics. The second property ensures that given $\sigma\colon \alpha \to \beta$ any bisimulation for $\alpha$ can be recovered by restricting some bisimulation for $\beta$ to the image of $\carr{\alpha}$ in $\carr{\beta}$ through the injection $\sigma^{c}$.

Fullness identifies reductions that use the state space of the target system in its entirety which means that full reductions do not introduce auxiliary states. 

\begin{proposition}
	\label[proposition]{thm:sys-reduction-fullness}
	For $\sigma\colon \alpha \to \beta$ a full system reduction, the map $\sigma^{c}$ is an isomorphism and the relation $\sigma^{b}$ is right-unique\footnote{A relation $R\subseteq X \times Y$ is called right-unique if for every $x \in X$ there is $y \in Y$ s.t.~$x \mathrel{R} y$.}.
\end{proposition}

\begin{proof}
	Assume that $\sigma$ is full. The function $\sigma^{c}$ is injective and surjective hence an isomorphism. Then it follows from condition \eqref{eq:reduction-bisim-condition} that $\sigma^{b}$ is also right-unique since states related by a bisimulation of the target system are always image of states of the source system
\end{proof} 

System reductions can be extended to categories of systems by equipping functors with them and ensuring they respect the structure of homomorphisms. Formally:
\begin{definition}
	\label[definition]{def:reduction}
	For $\cat{C}$ and $\cat{D}$ categories of system, a \emph{reduction} $\sigma$ from $\cat{C}$ to $\cat{D}$, written $\sigma\colon \cat{C} \to \cat{D}$, is a functor going from $\cat{C}$ to $\cat{D}$ equipped with a collection of system reductions
	\[
		\{\sigma_{\alpha}\colon \alpha \to \sigma(\alpha)\}_{\alpha \in \cat{C}}
	\]
	such that for any $f\colon \alpha \to \beta$ in $\cat{C}$:
	\[
		\sigma_{\beta}^{c} \circ f = \sigma(f) \circ \sigma_{\alpha}^{c}
		\text{.}
	\]
	A reduction $\sigma\colon \cat{C} \to \cat{D}$ is called \emph{full} if, and only if, every system reduction $\sigma_\alpha$ is full. A category $\cat{C}$ is said to \emph{reduce} (resp.~fully reduce) to $\cat{D}$, if there is a reduction (resp.~a full reduction) going from $\cat{C}$ to $\cat{D}$.
\end{definition}

\begin{notation}
	For categories $\cat{C}$ and $\cat{D}$ we write $\cat{C} \reducesTo \cat{D}$ if  $\cat{C}$ reduces to $\cat{D}$, $\cat{C} \reducesEq \cat{D}$ if $\cat{C} \reducesTo \cat{D}$ and $\cat{C} \reducesFrom \cat{D}$, $\cat{C} \freduceTo \cat{D}$ and $\cat{C} \freduceEq \cat{D}$ if the reductions involved are full.
\end{notation}

Reductions form a category.
To compose reductions $\sigma\colon \cat{C} \to \cat{D}$ and $\tau\colon \cat{D} \to \cat{E}$ it suffices to compose their underlying functors and each system reduction accordingly: the composite reduction $\tau \circ \sigma\colon \cat{C} \to \cat{E}$ is the composite functor $\tau \circ \sigma$ equipped with the family of system reductions given on each $\alpha \in \cat{C}$ by $(\tau \circ \sigma)_\alpha^{c} \defeq \tau_{\sigma(\alpha)}^{c} \circ \sigma_\alpha^{c}$ and $(\tau \circ \sigma)_\alpha^{b} \defeq \tau_{\sigma(\alpha)}^{b} \circ \sigma_\alpha^{b}$. Reduction composition is associative and admits identities which are given on every $\cat{C}$ as the identity assignments for systems and homomorphisms. Any reduction restricts to a reduction from a subcategory of its domain and extends to a reduction to a super-category of its codomain. Fullness is preserved by the above operations and full reductions form a category that lies in the category of reductions. 

\subsection{Reductions and final systems}
\label{sec:reductions-final-coalgebras}

Final systems describe all abstract behaviours for their type \ie all behaviours bisimulations cannot distinguish. 
Reductions preserve and reflect bisimulations and are injective on state spaces.
As a consequence, for a reduction going from $\Sys{S}$ to $\Sys{T}$ to exist, it is necessary that there are at least as many abstract behaviours of type $T$ as are those of type $S$.

\begin{lemma}
	\label[lemma]{thm:reductions-final-coalgebras-necc}
	For any $S$ and $T$, if both admit final systems then:
	\[
		\Sys{S} \reducesTo \Sys{T}
		\implies 
		\card{\carr{\ffix S}} \leq \card{\carr{\ffix T}}
		\text{.}
	\]
\end{lemma}

\begin{proof}
	Let $\sigma$ be a reduction going from $\Sys{S}$ to $\Sys{T}$.
	It follows from \cref{thm:sys-reduction-diamond} that $\card{\carr{\ffix S}} \leq \card{\carr{\sigma(\ffix S)}}$
	and from Condition \eqref{eq:reduction-bisim-condition} that $\sigma(\ffix S)$ is strongly extensional on the image of $\carr{\ffix S}$ \ie:
	\[
		R \in \bis(\sigma(\ffix S)) \implies (\sigma^{c}_{\ffix S}(x) \mathrel{R} \sigma^{c}_{\ffix S}(y) \iff x = y)
		\text{.}
	\]
	As a consequence, the $T$-system $\sigma(\ffix S)$ exhibits at least as many distinct abstract behaviours as the final $S$-systems $\ffix S$. Therefore, 
	$\card{\carr{\ffix S}} \leq \card{\carr{\ffix T}}$ or otherwise there would be $x, y \in \carr{\ffix S}$ such that $x \nsim_{S} y$ but $\sigma^{c}_{\ffix S}(x) \sim_{T} \sigma^{c}_{\ffix S}(y)$.
\end{proof}

The condition on the size of final systems is not only necessary but also sufficient for having a reduction. Intuitively, the size condition guarantees that it is possible to represent abstract behaviours of the source type $S$ as abstract behaviour of type $T$ by means of a system reduction for the final system. In fact, provided that $\card{\carr{\ffix S}} \leq \card{\carr{\ffix T}}$, it is possible to fix an injection $\sigma^{c}\colon \carr{\ffix S} \to \carr{\ffix T}$ to be used as an ``encoding''. Since final systems are strongly extensional, this map uniquely identifies a left-total relation $\sigma^{b} \subseteq \bis(\ffix S) \times \bis(\ffix T)$ such that
\[
	R \mathrel{\sigma^{b}} R'
	\iff
	(x \mathrel{R} y \iff \sigma^{c} (x) \mathrel{R'} \sigma^{c}(y))
	\text{.}
\]
By construction, $\sigma^{c}$ and $\sigma^{b}$ define a system reduction $\sigma\colon \ffix S \to \ffix T$. Once all abstract behaviours are covered, the reduction can be extended to every other $S$-system along their final semantics.

\begin{lemma}
	\label[lemma]{thm:reductions-final-coalgebras-suff}
	For any $S$ and $T$, if both admit final systems then:
	\[
		\Sys{S} \reducesTo \Sys{T}
		\impliedby
		\card{\carr{\ffix S}} \leq \card{\carr{\ffix T}}
		\text{.}
	\]
\end{lemma}

\begin{proof}
	Assume that $\card{\carr{\ffix S}} \leq \card{\carr{\ffix T}}$ and let $e\colon \carr{\ffix S} \to \carr{\ffix T}$ be any injective map.
	For $\alpha\colon X \to SX$ define $\sigma(\alpha)$ as the $T$-system depicted in the diagram below:
	\[
	\begin{tikzpicture}[
					font=\footnotesize,auto,
					xscale=2.8, yscale=1.5,
					baseline=(current bounding box.center),
			]
			\node (n1) at (0,1) {\(X + \carr{\ffix T}\)};
			\node (n2) at (.6,0) {\(\carr{\ffix S} + \carr{\ffix T}\)};
			\node (n3) at (2.2,0) {\(\carr{\ffix T}\)};
			\node (n4) at (3,0) {\(T\carr{\ffix T}\)};
			\node (n5) at (3.6,1) {\(T(X + \carr{\ffix T})\)};
			
			\draw[->] (n1) to node[swap] {\(!_{\alpha} + {id}_{\carr{\ffix T}}\)} (n2);
			\draw[->] (n2) to node[swap] {\([e,{id}_{\carr{\ffix T}}]\)} (n3);
			\draw[->] (n3) to node[swap] {\(\ffix T\)} (n4);
			\draw[->] (n4) to node[swap] {\(T({inr}_{X,\carr{\ffix T}})\)} (n5);
			\draw[->] (n1) to node {\(\sigma(\alpha)\)} (n5);
		\end{tikzpicture}
	\]
	where ${inr}_{X,\carr{\ffix T}}$ is the right injection into the coproduct $X + \carr{\ffix T}$.
	Let $R \in \bis(\alpha)$ and consider the relation $R' = R \cup \Delta_{\carr{\ffix T}}$ on $\carr{\sigma(\alpha)}$. Observe that $x \mathrel{R'} y$ whenever $x = y$ or $x \mathrel{R} y$ and that $R \subseteq {\sim_S}$. Therefore, the relation $R'$ is a bisimulation for $\sigma(\alpha)$ by construction.
	Define $\sigma_{\alpha}^{c}\colon X \to X + \carr{\ffix T}$ as the left injection into the coproduct and $\sigma_{\alpha}^{b}$ as the (left-total) relation $\{(R, R \cup \Delta_{\carr{\ffix T}}) \mid R \in \bis(\alpha)\}$. 
	If follows from the observation above that $\sigma_{\alpha}^{c}$ and $\sigma_{\alpha}^{b}$ form a reduction going from $\alpha$ to $\sigma(\alpha)$.	
	For $f\colon \alpha \to \beta$ a $S$-homomorphism, the function $f \times {id}_{\carr{\ffix T}}$ carries a $T$-homomorphism going from $\sigma(\alpha)$ to $\sigma(\beta)$ as shown by the commuting diagram below:
	\[
		\begin{tikzpicture}[
					font=\footnotesize,auto,
					xscale=2.8, yscale=1.5,
					baseline=(current bounding box.center),
			]
			\node (n1) at (0,1) {\(X + \carr{\ffix T}\)};
			\node (n2) at (.6,0) {\(\carr{\ffix S} + \carr{\ffix T}\)};
			\node (n3) at (2.2,0) {\(\carr{\ffix T}\)};
			\node (n4) at (3,0) {\(T\carr{\ffix T}\)};
			\node (n5) at (3.6,1) {\(T(X + \carr{\ffix T})\)};
			\node (n6) at (0,-1) {\(Y + \carr{\ffix T}\)};
			\node (n7) at (3.6,-1) {\(T(Y + \carr{\ffix T})\)};
			
			\draw[->] (n1) to node[] {\(!_{\alpha} + {id}_{\carr{\ffix T}}\)} (n2);
			\draw[->] (n1) to node[] {\(\sigma(\alpha)\)} (n5);
			\draw[->] (n2) to node[swap] {\([e,{id}_{\carr{\ffix T}}]\)} (n3);
			\draw[->] (n3) to node[swap] {\(\ffix T\)} (n4);
			\draw[->] (n4) to node[] {\(T({inr}_{X,\carr{\ffix T}})\)} (n5);
			\draw[->] (n4) to node[swap] {\(T({inr}_{Y,\carr{\ffix T}})\)} (n7);
			\draw[->] (n6) to node[swap] {\(!_{\beta} + {id}_{\carr{\ffix T}}\)} (n2);
			\draw[->] (n6) to node[swap] {\(\sigma(\beta)\)} (n7);
			\draw[->] (n1) to node[swap] {\(f \times {id}_{\carr{\ffix T}}\)} (n6);
			\draw[->] (n5) to node[] {\(T(f \times {id}_{\carr{\ffix T}})\)} (n7);
			
		\end{tikzpicture}
	\]
	Note that the left triangle commutes by definition of final system.
	As a consequence, the mappings $\alpha \mapsto \sigma(\alpha)$ and $f \mapsto f \times {id}_{\carr{\ffix T}}$ define a functor $\sigma\colon \Sys{S} \to \Sys{T}$. This functor is coherent with the family of system reductions $\{\sigma_{\alpha}\}_{\alpha \in \Sys{S}}$ defined above and hence extends to a reduction from $\Sys{S}$ to $\Sys{T}$.	
\end{proof}

It follows from \cref{thm:reductions-final-coalgebras-necc,thm:reductions-final-coalgebras-suff} that the condition on the size of final systems is both necessary and sufficient for the existence of a reduction. Formally:

\begin{theorem}
	\label[theorem]{thm:reductions-final-coalgebras}
	For any $S$ and $T$, if both admit final systems then:
	\[
		\Sys{S} \reducesTo \Sys{T}
		\iff
		\card{\carr{\ffix S}} \leq \card{\carr{\ffix T}}
		\text{.}
	\]
\end{theorem}

\noindent
As a consequence, to prove or disprove that there is a reduction taking $S$-systems to $T$-systems it suffices to provide an upper and a lower bound to the size of their final systems, respectively.

\section{Reducing weighted transition systems}
\label{sec:reducing-wlts}

In this section we address the main question left open in \cite{mp:ictcs2016} \ie whether it is possible to collapse the hierarchy of FuTSs beyond the category of WLTSs. 
To answer this question, we determine conditions on subcategories of $\cat{WLTS}$ that are necessary and sufficient for a reduction to exist. Remarkably, these conditions involve solely the total number of labels and weights used by systems of a certain type. This suggests the possibility for reductions to shift information between labels and weights while preserving the semantics of WLTSs. 

In \cite{klin:sas2009}, Klin proved that every category $\cat{WLTS}(A,M)$ has final objects but to the best of our knowledge the cardinality of their carriers has not been studied yet. 
We determine lower and upper bounds to said cardinality. Albeit conservative, these bounds have the advantage of offering a clear dependency on the cardinality of labels and weights.
For any non-empty set $A$ and non-trivial abelian monoid $M$, it holds that:
\begin{equation}
	\label{eq:wlts-final-coalgebras-lower-bound}
	\tag{$\spadesuit$}
	\max(\aleph_0,\card{M},\card{A}) \leq \card{\carr{\ffix (\ff{M}-)^A}}
	\text{.}
\end{equation}
Intuitively, the type $(\ff{M}-)^A$ has at least
\begin{itemize}
\item $\aleph_0$-many abstract behaviours since $M$ is not trivial and natural numbers can be encoded using computation length;
\item $\card{M}$-many abstract behaviours since there is at least a computation whose first step assigns a given weight to a state;
\item $\card{A}$-many abstract behaviours since there is at least a computation that is able to proceed under a given label.
\end{itemize}
For any non-empty set $A$ and non-trivial abelian monoid $M$, it holds that:
\begin{equation}
	\label{eq:wlts-final-coalgebras-upper-bound}
	\tag{$\clubsuit$}
	\card{\carr{\ffix (\ff{M}-)^A}} \leq \card{M}^{\max(\aleph_0,\card{A})}
	\text{.}
\end{equation}
Intuitively, abstract behaviours for $(\ff{M}-)^A$ can be pictured as certain possibly infinite trees that alternate $A$-indexed branches and  $M$-labelled finite branches. These trees have at most  $(\card{M}^{\aleph_0})^{\card{A}} = \card{M}^{\aleph_0 \cdot \card{A}} = \card{M}^{\max(\aleph_0,\card{A})}$ possible ways of branching and their depth is countable. Thus, there are not more than $\left(\card{M}^{\max(\aleph_0,\card{A})}\right)^{\aleph_0} = \card{M}^{\max(\aleph_0,\card{A})}$ of such trees. 

\begin{lemma}
	\label[lemma]{thm:wlts-final-coalgebras-bounds}
	For $(M,+,0)$ a non-trivial abelian monoid and $A$ a non-empty set, $(\ff{M}-)^A$ has final systems and the cardinality of their carrier satisfies \eqref{eq:wlts-final-coalgebras-lower-bound} and \eqref{eq:wlts-final-coalgebras-upper-bound}.
\end{lemma}

\begin{proof}
	Recall from \cite[Prop.~3]{klin:sas2009} that $(\ff{M}-)^A$ has final coalgebras. 
		
	In order to prove that \eqref{eq:wlts-final-coalgebras-lower-bound} holds, we claim that the following implication holds:
	\[
		X \cong (\ff{M}X)^A \implies \card{X} \geq \max(\aleph_0,\card{M},\card{A})
		\text{.}
	\]
	Then, \eqref{eq:wlts-final-coalgebras-lower-bound} follows from Lambeck's Lemma---which states that final coalgebras are always isomorphisms \ie invariants of their endofunctor. To prove our claim we proceed considering  three main cases: one for each factor in $\max(\aleph_0,\card{M},\card{A})$.
	For the first case assume that $\aleph_0 = \max(\aleph_0,\card{M},\card{A})$.  
	Assume by contradiction that $X$ is finite. As a consequence, $(\ff{M}X)^A$ is the function space $\left(M^X\right)^A$ leading to the contradiction:
	\[
		\card{X}= \card{(\ff{M}X)^A} = \card{M}^{\card{X}\card{A}} \geq 2^{\card{X}} > \card{X}
	\]
	since $\card{X} < \aleph_0$, $\card{M} \geq 2$, and $\card{A} \geq 1$. Therefore, $\card{X} \geq \aleph_0$ must follow from the premise that $X \cong (\ff{M}X)^A$ and the assumption that $\aleph_0 = \max(\aleph_0,\card{M},\card{A})$. 
	For the second case assume that $\card{M} = \max(\aleph_0,\card{M},\card{A})$. 
	Fix a state $x \in X$. It follows that:
	\[
		\card{M} = \card{\{x \cdot m \mid m \in M\}} \leq \card{(\ff{M}X)^A} = \card{X}
	\]
	since $\lambda a. x \cdot m \in (\ff{M}X)^A$.
	For the second case assume that $\card{A} = \max(\aleph_0,\card{M},\card{A})$. Fix any $x \in X$ and any $m \in M \setminus \{0\}$. For $a \in A$ define $\delta_a\colon A \to \ff{M}X$ as the Dirac's delta function taking $a$ to $x \cdot m$ and everything else to the constantly zero function. It holds that:
	\[
		\card{A} = \card{\{\delta_a \mid a \in A\}} \leq \card{(\ff{M}X)^A} = \card{X}\text{.}
	\]
	
	Recall that the final sequence for $(\ff{M}-)^A$ is a limit preserving functor $F\colon \cat{Ord} \to \Set$ from the category of ordinal numbers to that of sets and such that:
	\[
		F_{\alpha+1} = (\ff{M}F_{\alpha})^A
		\qquad
		F^{\beta+1}_{\alpha+1} = (\ff{M}F^{\beta}_{\alpha})^A
	\]
	where $\alpha \leq \beta$ are ordinal numbers and $F^{\beta}_{\alpha}\colon F_\beta \to F_\alpha$ is the action induced by $\alpha \leq \beta$.	
	We claim that for any ordinal $\alpha$,
	\[
		\card{F_{\alpha}} \leq \card{M}^{\max(\aleph_0,\card{A})}
		\text{.}
	\]
	Then, \eqref{eq:wlts-final-coalgebras-upper-bound} holds since the the final sequence for $(\ff{M}-)^A$ stabilises at the final $(\ff{M}-)^A$-coalgebra.
	\setupstepsequence{thm:wlts-final-coalgebras-bounds}%
	Let $\alpha$ be $0$. $F_0 = \varprojlim_{\alpha < 0} F^{0}_{\alpha} = 1$.
	Let $\alpha$ be $n+1$ for $n < \omega$. It holds that:
	\begin{align*}
		\card{F_{n+1}}
		& \steplabel{=} 
		\left(\sum_{k \in \mathbb{N}} \card{M\setminus\{0\}}^k \cdot \card{F_{n}}^k\right)^{\card{A}}
		\steplabel{\leq}
		\left(\sum_{k \in \mathbb{N}} \card{M}^k \cdot \left(\card{M}^{\max(\aleph_0,\card{A})}\right)^k\right)^{\card{A}}
		\\ 
		&=
		\left(\sum_{k \in \mathbb{N}} \left(\card{M}^{\max(\aleph_0,\card{A})}\right)^k\right)^{\card{A}}
		\leq
		\left(\card{M}^{\max(\aleph_0,\card{A})}\right)^{\aleph_0\cdot\card{A}}
		=
		\card{M}^{\max(\aleph_0,\card{A})}
		\text{.}
	\end{align*}
	where \stepref{1}	follows from the observation that 
	$\ff{M}X = \bigcup_{k \in \mathbb{N}} \{ \rho \in \ff{M}X \mid \card{\supp(\rho)} = k \}$, \stepref{2}	from the inductive hypothesis, and the remaining inequalities from basic cardinal arithmetic.	
	Let $\alpha$ be $\omega$.	Observe that the limit step	$F_{\omega}$ is a subset of $\prod_{n<\omega} F_{n}$. Therefore, it holds that: 
	\begin{align*}
		\card{F_{\omega}}
		& \leq 
		\prod_{n < \omega} \card{F_{n}}
		\leq 
		\prod_{n < \omega} \card{M}^{\max(\aleph_0,\card{A})}
		\leq
		\left(\card{M}^{\max(\aleph_0,\card{A})}\right)^{\aleph_0}
		\leq
		\card{M}^{\max(\aleph_0,\card{A})}
		\text{.}
	\end{align*}
	We claim that for any transfinite ordinal $\alpha$ it holds that $\card{F_{\alpha}} \leq \card{F_{\omega}}$ and, as a consequence of the above, \eqref{eq:wlts-final-coalgebras-upper-bound}. Observe that $\ff{M}$ is finitary and that $(\ff{M}-)^A$ preserves injections. Recall from \cite{worrell:tcs2005} that final sequences for finitary endofunctors stabilise after $\omega + \omega$ steps (\cf \cite[Thm.~11]{worrell:tcs2005}), that all actions for $\omega < \alpha$ are injections, and that this holds also for arbitrary products of finitary endofunctors (\cf \cite[Thm. 13]{worrell:tcs2005}).
	We conclude that our claims holds true.
\end{proof}

Recall from \cref{thm:sys-reduction-diamond} that reductions never take systems to systems with smaller state spaces and from \cref{thm:reductions-final-coalgebras} that this situation applies also to abstract behaviours described by final systems. Therefore, a consequence of the dependency shown above between the size of final WLTSs and the number of labels and weights they use is that the target of a reduction from the category $\cat{WLTS}$ to any of its subcategories must have systems with enough labels and weights for every small\footnote{A cardinal number is called \emph{small} if it is the cardinality of a (small) set as opposed to proper classes.} cardinal. Formally:

\begin{theorem}
	\label[theorem]{thm:wlts-reduction-necc}
	For $\cat{C}$ a subcategory of $\cat{WLTS}$, if $\cat{WLTS} \reducesTo  \cat{C}$ then for any small cardinal $\kappa$ there is $\alpha\colon X \to (\ff{M}X)^A$ in $\cat{C}$ such that $\kappa \leq \max(\card{M},\card{A})$.
\end{theorem}

\begin{proof}
	Let $\sigma$ be a reduction going from $\cat{WLTS}$ to $\cat{C}$. 
	Fix a small cardinal $\kappa$ and assume without loss of generality that it is transfinite.
	Let $A$ and $(M,+,0)$ be a non-empty set and a non-trivial abelian monoid with the property that $2^\kappa < \card{A} + \card{M}$. Write $\underline{A}$ and $\underline{M}$ for the set of labels and monoid of weights defining the type of the system $\sigma(\ffix (\ff{M}-)^A)$. Observe that $\sigma$ restricts to a reduction going from $\Sys{(\ff{M}-)^{A}}$ to $\Sys{(\ff{\underline{M}}-)^{\underline{A}}}$ since homomorphisms respect types and reductions are functors.
	\setupstepsequence{thm:wlts-reduction-1}%
	Then, it holds that:
	\begin{align*}
		2^\kappa 
		& \steplabel{<}
		\max(\card{M},\card{A}) 
		\steplabel{\leq}
		\card{\carr{\ffix (\ff{M}-)^{A}}}
		\steplabel{\leq}
		\card{\carr{\sigma(\ffix (\ff{M}-)^A)}}
		\steplabel{\leq}
		\card{\carr{\ffix (\ff{\underline{M}}-)^{\underline{A}}}}
		\\ 
		&\steplabel{\leq}
		\card{\underline{M}}^{\max\left(\aleph_0,\card{\underline{A}}\right)}
	\end{align*}
	where \stepref{1} follows from the assumption that $\kappa \geq \aleph_0$ and definition of cardinal sum, \stepref{2} and \stepref{4} follow from \cref{thm:wlts-final-coalgebras-bounds}, \stepref{3} follows from \cref{thm:sys-reduction-diamond}, and \stepref{5} follows from the same reasoning used in the proof of \cref{thm:reductions-final-coalgebras}.
	We claim that $\kappa \leq \max(\card{\underline{M}},\card{\underline{A}})$.
	The proof is divided in three cases.
	For the first case, assume that $\card{\underline{M}} < \aleph_0$.
	As a consequence, it holds that: 
	\[
		\card{\underline{M}}^{\max\left(\aleph_0,\card{\underline{A}}\right)}
		=
		2^{\max\left(\aleph_0,\card{\underline{A}}\right)}
		=
		\max\left(2^{\aleph_0},2^{\card{\underline{A}}}\right)
		\text{.}
	\]
	Observe that $2^{\aleph_0} \leq 2^{\kappa}$ since cardinal exponentiation is monotonic and $\aleph_0 \leq \kappa$ by assumption. As a consequence of the last of observation, since the above inequality holds we conclude that $2^{\kappa} < 2^{\card{\underline{A}}}$ and hence that $\kappa \leq \card{\underline{A}}$. Therefore, $\kappa \leq \card{\underline{A}} + \card{\underline{M}}$.	
	For the second case, assume that $\card{\underline{M}} \geq \aleph_0$ and that $\card{\underline{A}} \geq \aleph_0$. As a consequence, it holds that:  
	\[
		\card{\underline{M}}^{\max\left(\aleph_0,\card{\underline{A}}\right)} 
		= 
		\card{\underline{M}}^{\card{\underline{A}}} 
		\geq 
		2^{\card{\underline{A}}}
		\text{.}
	\]
	Then, the claim follows from the same argument discussed in the previous case.	
	For the last case, assume that $\card{\underline{M}} \geq \aleph_0$ and $\card{\underline{A}} \leq \aleph_0$. As a consequence, it holds that: 
	\[
		2^{\kappa} 
		< 
		\card{\underline{M}}^{\max(\aleph_0,\card{\underline{A}})} 
		=
		\card{\underline{M}}^{\aleph_0}
		\leq 
		\max\left(2^{\card{\underline{M}}},2^{\aleph_0}\right)
		= 
		2^{\card{\underline{M}}}
	\]
	where the last inequality is an instance of a general result of powers involving at least a transfinite cardinal.
	It follows by monotonicity of cardinal exponentiation that $\kappa \leq \card{\underline{M}}$ hence the claim.
	Since the only assumption put on $\kappa$ is that it is transfinite, the thesis holds true.
\end{proof}

A consequence of \cref{thm:wlts-reduction-necc} is that it cannot exist an instance of the WLTS meta-model that is ``universal'': it is not possible to faithfully express all behaviours modelled as WLTSs using a single set of labels and a single monoid of weights. 

\begin{corollary}
	\label[corollary]{thm:wlts-no-universal}
	There are no $A$ and $M$ such that $\cat{WLTS} \reducesTo \cat{WLTS}(A,M)$.
\end{corollary}

\begin{proof}
 	Assume by contradiction that there are $A$ and $M$ such that $\cat{WLTS} \reducesTo\cat{WLTS}(A,M)$. It follows from \cref{thm:wlts-reduction-necc} that $\card{A} + \card{M}$ is greater than any small cardinal. As a consequence, at least one between $A$ and $M$ is a proper class whereas they are sets by definition of WLTS.
\end{proof}

\Cref{thm:wlts-reduction-necc} describes a condition on a subcategory $\cat{C}$ of $\cat{WLTS}$ that is necessary to reduce $\cat{WLTS}$ to $\cat{C}$. However, said condition fails to be sufficient. For a counter example, consider as $\cat{C}$ the category of all final weighted transition systems. This category satisfies the hypothesis of \cref{thm:wlts-reduction-necc}: for any small cardinal $\kappa$ the category has a system whose type has at least $\kappa$-many labels and weights, combined. However, it is not possible to reduce $\cat{WLTS}$ to $\cat{C}$ for the same reason that it is not possible to reduce $\cat{WLTS}(A,M)$ to its subcategory of final systems: reductions are injective on state spaces and respect system homomorphisms. Because of the injectivity property the target category must at least have, for any small cardinal $\kappa$ a system with a carrier exceeding $\kappa$. Because homomorphisms preserve types, systems from $\cat{WLTS}(A,M)$ must be reduced to systems of the same type.
These observations suggest the following sufficient condition which, roughly speaking, extends that of \cref{thm:wlts-reduction-necc} by assuming that $\cat{C}$ has also enough systems for each type in the reduction targeted.

\begin{theorem}
	\label[theorem]{thm:wlts-reduction-suff}
	For $\cat{C}$ a category of systems, if for any small cardinal $\kappa$ there are $A$ and $M$ such that $\kappa \leq \max(\card{M},\card{A})$ and $\cat{WLTS}(A,M) \reducesTo \cat{C}$ then, $\cat{WLTS}\reducesTo \cat{C}$.
\end{theorem}

\begin{proof}
	Observe that system homomorphisms preserve types and reductions respect homomorphisms. As a consequence, any subcategory $\cat{WLTS}(A,M)$ of $\cat{WLTS}$ can be considered and reduced in isolation: to reduce $\cat{WLTS}$ to $\cat{C}$ it is sufficient to provide a reduction for each $\cat{WLTS}(A,M)$.
	Given any $A$ and $M$, let $\underline{A}$ and $\underline{M}$ be a non-empty set and a non-trivial abelian monoid with the property that $\card{\carr{\ffix (\ff{M}-)^{A}}} \leq \card{\carr{\ffix (\ff{\underline{M}}-)^{\underline{A}}}}$ and $\cat{WLTS}(\underline{A},\underline{M}) \reducesTo \cat{C}$.
	Existence of $\underline{A}$ and $\underline{M}$ follows from hypothesis on $\cat{C}$ and \cref{thm:wlts-final-coalgebras-bounds}: if  $\card{M}^{\max(\aleph_0,\card{A})}$ is taken as $\kappa$ then it holds that:
	\[
		\card{\carr{\ffix (\ff{M}-)^{A}}} \leq
		\card{M}^{\max(\aleph_0,\card{A})} \leq 
		\max(\underline{A},\underline{M}) \leq 
		\card{\carr{\ffix (\ff{\underline{M}}-)^{\underline{A}}}}
		\text{.}
	\]
	It follows from \cref{thm:reductions-final-coalgebras} and hypothesis on $\underline{A}$ and $\underline{M}$ that $\cat{WLTS}(A,M) \reducesTo \cat{WLTS}(\underline{A},\underline{M}) \reducesTo \cat{C}$ proving that any subcategory $\cat{WLTS}(A,M)$ of $\cat{WLTS}$ reduces to $\cat{C}$ and hence that $\cat{WLTS} \reducesTo \cat{C}$.
\end{proof}

\looseness=-1
Neither \cref{thm:wlts-reduction-necc} nor \cref{thm:wlts-reduction-suff}
distinguish between the contribution of labels and weights. Instead, both require that any small cardinal is covered by the combined size of labels and weights. This situation points to a trade-off between labels and weights: it is possible to shift information and complexity between the two parameters defining WLTSs behaviours while preserving their semantics in terms of quantitative bisimulation. At one end of this trade-off all information is shifted into labels and weights are drawn from the smallest non-trivial abelian monoid which is (up-to isomorphism) the monoid $\mathbb{B}$ of boolean values under disjunction. Systems weighted over $\mathbb{B}$ are LTSs.

\begin{corollary}
	\label[corollary]{thm:wlts-reduce-to-lts}
	$\cat{WLTS} \reducesEq \cat{LTS}$.
\end{corollary}

\noindent At the other end of the trade-off all information is expressed using weights only. These systems have exactly one label which means that they are essentially ``unlabelled''. We write $\cat{WTS}$ for the subcategory of $\cat{WLTS}$ formed by ``unlabelled'' weighted transition systems.

\begin{corollary}
	\label[corollary]{thm:wlts-reduce-to-wts}
	$\cat{WLTS} \reducesEq \cat{WTS}$.
\end{corollary}

Actually, a result stronger than that stated in \cref{thm:wlts-reduce-to-wts} holds: there is a full reduction. 

\begin{theorem}
	\label[theorem]{thm:wlts-fully-reduce-to wts}
		$\cat{WLTS} \freduceEq \cat{WTS}$.
\end{theorem}

\begin{proof}
	For $(M,+,0)$ an abelian monoid and $A$ a set, the function space $M^A$ carries an abelian monoid structure given by an $A$-indexed product of monoids. In particular, the monoidal zero is the $A$-indexed tuple: 
	\[
		\vv{0} = \sum_{a \in A} 0 \cdot a
	\]
	and monoidal sum is defined on each $\vv{m}$ and $\vv{n}$ as:
	\[
		\sum_{a \in A} m_{a} \cdot a + \sum_{a \in A} n_a \cdot a = 
		\sum_{a \in A} (m_a + n_a) \cdot a
		\text{.}
	\]
	There is a natural isomorphism $\mu\colon (\ff{M}-)^A \cong \ff{M^A}$
	given on each set $X$ as:
	\[
		\mu_{X}(\phi) = \sum_{x \in X} \lambda a \in A.\phi(a)(x) \cdot x
	\]
	and
	\[
		\mu_{X}^{-1}(\psi) = \lambda a \in A. \sum_{x \in X} \psi(x)(a) \cdot x
		\text{.}
	\]
	Since this natural transformation is component-wise injective, it follows from \cite[Thm.~3]{mp:ictcs2016} that there is a full reduction $\hat{\mu}\colon \cat{WLTS}(A,M) \to \cat{WLTS}(1,M^A)$ given
	on each system $\alpha$ as:
	\[
		\hat{\mu}(\alpha) \defeq \mu_{\carr{\alpha}} \circ \alpha
		\qquad
		\hat{\mu}_{\alpha}^{c} \defeq id_{\carr{\alpha}}
		\qquad
		\hat{\mu}_{\alpha}^{b} \defeq \Delta_{\bis(\alpha)}
	\]
	and each homomorphism $f\colon \alpha \to \beta$ as follows:
	\[
		\hat{\mu}(f) \defeq f
		\text{.}
		\qedhere
	\]
\end{proof}

\section{Reducing state-to-function transition systems, again}
\label{sec:reducing-futs}

In \cite[Sec.~6]{mp:ictcs2016,mp:arxiv2017} we proved that the category of FuTSs reduces to that of WLTSs by defining a suitable reduction. An explicit reduction offered us a bridge for deriving logical characterisations of bisimulations for these systems \cite[Sec.~8]{mp:arxiv2017} but this may not be necessary in general. In this section we present an alternative (and shorter) proof based on the results we introduced in \cref{sec:reductions-final-coalgebras}. In this way, we are also able to validate the technique rediscovering known reductions.

Recall from \cref{sec:reductions-final-coalgebras} that to prove existence of reductions using \cref{thm:reductions-final-coalgebras} it suffices to provide an upper bound for the size of final systems for the source type (here FuTSs) and a lower bound for the size of final systems for the target type (here WLTSs).

\begin{lemma}
	\label[lemma]{thm:futs-final-coalgebras-bound}
	There are final $\ffix (\ff{\vv{M}}-)^{\vv{A}}$-systems and it holds that:
	\[
		\card{\carr{\ffix (\ff{\vv{M}}-)^{\vv{A}}}} \leq \max_{M \in \vv{M}}(\card{M})^{\max_{A \in \vv{A}}(\aleph_0,\card{A})}
		\text{.}
	\]
\end{lemma}

\begin{proof}
	For any non-trivial abelian monoid $M$, set $X$ and small cardinal $\kappa > \aleph_0$ it holds that:
	\begin{equation*}
		\max\left(\card{M},\card{X}\right) \leq \kappa \implies \card{\ff{M}X} \leq \kappa
	\end{equation*}
	In fact, if $\max\left(\card{M},\card{X}\right) \leq \kappa $ then,  
	\[
		\card{\ff{M}X} \leq \sum_{h < \omega}\card{M}^h \cdot \card{X}^h \leq \sum_{h < \omega} \kappa^h = \kappa
		\text{.}
	\]
	Write $F\colon \cat{Ord} \to \Set$ for the final sequence of the endofunctor $(\ff{\vv{M}}-)^{\vv{A}}$ and $\mu$ for the cardinal number $\max_{M \in \vv{M}}(\card{M})^{\max_{A \in \vv{A}}(\aleph_0,\card{A})}$. 
	The proof proceeds by showing that $\card{F_\alpha} \leq \mu$ for any ordinal $\alpha$. Observe that like $(\ff{M}-)^{A}$, also the final sequence for $(\ff{\vv{M}}-)^{\vv{A}}$ stabilises in $\omega+\omega$ and that all actions after $\omega$ are injections. It suffices to prove that $\card{F_\alpha} \leq \mu$ for any $\alpha \leq \omega$.
	Let $\alpha$ be $0$. By definition of $F$ it holds that $\card{F_0} = 1 < \mu$.
	Let $\alpha$ be a finite successor ordinal $\beta + 1$. It follows from the inductive hypothesis, the implication above, and definition of $F_{\beta +1}$ that:
	\[
		\card{F_{\beta+1}} 
		= 
		\card{\left(\ff{\vv{M}}F_\beta\right)^{\vv{A}}}
		=
		\card{\prod_{i = 0}^{n}\left(\ff{\vv*{M}{i}}F_\beta\right)^{A_i}}
		=
		\prod_{i = 0}^{n}\card{\ff{\vv*{M}{i}}F_\beta}^{\card{A_i}}
		\leq
		\prod_{i = 0}^{n}\mu^{\card{A_i}}
		\leq
		\mu
		\text{.}
	\]
	Let $\alpha$ be $\omega$. From definition of $F$ on limit ordinals and the inductive hypothesis it holds that:
	\[
		\card{F_\omega} = \card{\varprojlim_{\beta < \omega} F^{\omega}_{\beta}} \leq \card{\prod_{\beta < \omega} F_\beta} \leq \mu^{\aleph_0} = \mu
		\text{.}
	\]
	Therefore, every set in the final sequence for the endofunctor $(\ff{\vv{M}}-)^{\vv{A}}$ is of cardinality less or equal to $\mu = \max_{M \in \vv{M}}(\card{M})^{\max_{A \in \vv{A}}(\aleph_0,\card{A})}$, especially the carrier of $\ffix(\ff{\vv{M}}-)^{\vv{A}}$.
\end{proof}

The category of FuTSs reduces to that of WLTSs.

\begin{theorem}
	\label[theorem]{thm:futs-reduce-to-lts}
	$\cat{FuTS} \reducesEq \cat{WLTS}$.
\end{theorem}

\begin{proof}
	Let $\vv{A}$ and $\vv{M}$ be sequences of labels and monoids from a type modelling FuTSs.
	It follows from \cref{thm:futs-final-coalgebras-bound} that there is a non-empty set $\underline{A}$ and a non-trivial abelian monoid $\underline{M}$ such that:
	\[
		\card{\carr{\ffix (\ff{\vv{M}}-)^{\vv{A}}}} \leq \max(\underline{M},\underline{A})
		\text{.}
	\]
	Recall from \cref{thm:wlts-final-coalgebras-bounds} that $\max(\underline{M},\underline{A})$ is a lower bound to the size of final $(\ff{\underline{M}}-)^{\underline{A}}\mspace{1mu}$-systems, then $\cat{FuTS}(\vv{A},\vv{M}) \reducesTo \cat{WLTS}(\underline{A},\underline{M})$ by \cref{thm:reductions-final-coalgebras}.
	From the fact that every subcategory $\cat{FuTS}(\vv{A},\vv{M})$ of $\cat{FuTS}$ reduces to some subcategory $\cat{WLTS}(\underline{A},\underline{M})$ of $\cat{WLTS}$ it follows that $\cat{FuTS} \reducesTo \cat{WLTS}$.
\end{proof}

\section{Conclusions and final remarks}
\label{sec:conclusions}

In this paper we studied the category of weighted labelled transitions systems  under the lens of reductions in order to understand their expressiveness in terms of quantitative bisimulation. We proved that there is no instance of the WLTS model that is ``universal'': it is not possible to faithfully express all behaviours modelled as WLTSs using a single set of labels and a single monoid of weights (\cref{thm:wlts-no-universal}). In fact, we proved that the expressiveness of these models depends on combined size of their sets of labels and weights (the parameters of WLTSs models) which can be arbitrarily large sets but not proper classes (\cref{thm:wlts-reduction-necc}). Given that the dependency is on the combined size, we then proved that it is indeed possible to freely distribute the contribution between labels and weights while preserving the model expressiveness in terms of quantitative bisimulation (\cref{thm:wlts-reduction-suff}). At the two ends of this trade-off between labels and weights there are the categories of LTSs and of WTSs (or unlabelled WLTSs) where the information is completely shifted to labels or to weights, respectively (\cref{thm:wlts-reduce-to-lts,thm:wlts-reduce-to-wts}).

To obtain these results, we extended the framework of reductions with new results relating the existence of reductions to properties of final coalgebras. In particular, we proved that reductions exist if and only if the target type has at least as many abstract behaviours as the source type. As a consequence, even a rough estimate of the final coalgebra size suffices to prove or disprove the existence of a reduction. As an application, in \cref{sec:reducing-futs} we provided a shorter proof of a known result: that the category of FuTSs reduces to that of WLTSs.

\looseness=-1
Besides the classification interest, results presented in this work are of relevance to the extension of these theories of models. In fact, reductions pave the way for porting existing and new results between categories of transition systems. For instance, in \cite{mp:arxiv2017} reductions are used to derive new Hennessy-Milner-style modal logics for transition systems and guarantee an important property of such logics called full-abstraction (or expressiveness) \ie their ability to characterise bisimilarity.

In this (and previous works \cite{mp:ictcs2016,mp:arxiv2017}) we focused on strong bisimulation; a natural direction to pursue is to consider different behavioural equivalences, like trace equivalence or weak bisimulation. We remark that, as shown in \cite{hjs:lmcs2007,bmp:jlamp2015,bp:concur2016,peressotti:phdthesis}, in order to deal with these and similar equivalences, endofunctors need to be endowed with a monad (sub)structure; although WLTSs are covered in \cite{mp:arxiv2013,bmp:jlamp2015}, an analogous account of ULTraSs and FuTSs is still an open problem.

\paragraph*{Acknowledgements} This work was supported by the CRC project, grant no.~DFF–4005-00304 from the Danish Council for Independent Research, and by the Open Data Framework project at the University of Southern Denmark.

\printbibliography

\end{document}